\documentclass[11pt]{article}
\usepackage{amsmath, amssymb, amsthm}
\usepackage{latexsym}
\usepackage{color}
\usepackage{mathtools}
\usepackage{amsfonts}
\usepackage{hyperref}

\newtheorem{thm}{Theorem}[section]

\newtheorem{cor}[thm]{Corollary}

\newtheorem{claim}[thm]{Claim}

\newtheorem{theorem}[thm]{Theorem}

\usepackage{fullpage}

\newcommand{\RR}{\ensuremath{\mathbb{R}}}
\newcommand{\NN}{\ensuremath{\mathbb{N}}}

\newcommand{\E}{\mathbb{E}}
\newcommand{\suchthat}{\mathrel{:}}

\newcommand{\cone}{\ensuremath{\mathsf{cone}}}
\newcommand{\dotp}[2]{\langle #1, #2 \rangle}
\newcommand{\eps}{\epsilon}
\newcommand{\conv}{\operatorname{conv}}
\newcommand{\calS}{\ensuremath{\mathcal{S}}}

\def\final{1}  % set this to 1 to get a comment-free version

\ifnum\final=0  %namely if we allow comments in the output
\newcommand{\lnote}[1]{[{\small Luis: \bf #1}]\marginpar{*}}
\newcommand{\nnote}[1]{[{\small Navin: \bf #1}]\marginpar{*}}
\newcommand{\snote}[1]{[{\small Santosh: \bf #1}]\marginpar{*}}
\newcommand{\sidecomment}[1]{\marginpar{\tiny #1}}
\else % in this case [final=1] we don't want any comments to show
\newcommand{\lnote}[1]{}
\newcommand{\nnote}[1]{}
\newcommand{\snote}[1]{}
\newcommand{\sidecomment}[1]{}
\fi  %ok, we're done with defining comment macros
\newcommand{\email}[1]{\href{mailto:#1}{\texttt{#1}}}

\newcommand{\floor}[1]{\left\lfloor#1\right\rfloor}
\newcommand{\card}[1]{\lvert#1\rvert}
\newcommand{\dimd}{{d}}
\newcommand{\dimdm}{{d-1}}

\title{Lower Bounds for the Average and Smoothed Number of Pareto Optima}
\date{}

\author{Navin Goyal\\
Microsoft Research India \\
\email{navingo@microsoft.com}
\and
Luis Rademacher \\
Computer Science and Engineering \\
Ohio State University \\
\email{lrademac@cse.ohio-state.edu}}

\begin{document}

\maketitle
\begin{abstract}
%% Multiobjective optimization and Pareto optimality are fundamental tools in economics and several other
%% disciplines. A very useful methodology to solve multiobjective optimization problems is to generate the
%% set of

Smoothed analysis of multiobjective $0$--$1$ linear optimization has drawn considerable attention recently. The number of Pareto-optimal solutions (i.e., solutions with the property that no other solution is at least as good in all the coordinates and better in at least one) for multiobjective optimization problems is the central object of study. In this paper, we prove several lower bounds for the expected number of Pareto optima. Our basic result is a lower bound of $\Omega_d(n^{d-1})$ for optimization problems with $d$ objectives and $n$ variables under fairly general conditions on the distributions of the linear objectives. Our proof
relates the problem of lower bounding the number of Pareto optima to results in geometry connected to arrangements of hyperplanes.
We use our basic result to derive (1) To our knowledge, the first lower bound for natural multiobjective optimization problems. We
illustrate this for the maximum spanning tree problem with randomly chosen edge weights. Our technique is sufficiently flexible to yield such lower bounds for other standard objective functions studied in this setting (such as,  multiobjective shortest path, TSP tour, matching). (2) Smoothed lower bound of $\min \{ \Omega_d( n^{d-1.5} \phi^{(d-\log d) (1-\Theta(1/\phi))}), 2^{\Theta(n)} \}$ for the $0$--$1$ knapsack problem with $d$ profits for $\phi$-semirandom distributions for a version of the knapsack problem. This improves the recent lower bound of Brunsch and R{\"o}glin.

\end{abstract}

%\newpage

\section{Introduction}

Multiobjective optimization involves scenarios where there is more than one objective function to optimize: When planning a train trip we may want to choose connections that minimize fare, total time, number of train changes, etc.  The objectives may conflict with each other and there may not be a single best solution to the problem. Such multiobjective optimization problems arise in diverse fields ranging from economics to computer science, and have been well-studied. A number of approaches exist in the literature to deal with the trade-offs among the objectives in such situations: Goal programming, multiobjective approximation algorithms, Pareto-optimality; see, e.g., \cite{GrandoniRS09, MoitraO11, RoglinT09} for references. It is the latter approach using Pareto-optimality that concerns us in this paper. A Pareto-optimal solution is a solution with the property that no other solution is at least as good in all the objectives and better in at least one. Clearly, the set of
Pareto-optimal solutions (Pareto set in short) contains all desirable solutions as any other solution is
strictly worse than a solution in the Pareto set. In the worst case, the Pareto set can be exponentially large as a function of the input size (see, e.g., \cite{Muller-HannemannW01}). However, in many cases of interest, the Pareto set is typically not too large. Thus, if the Pareto set is small and can be generated efficiently, then it can be treated possibly with human assistance to choose among the few alternatives. Pareto sets are also used in heuristics for optimization problems (e.g. \cite{NemhauserU69}). To explain why Pareto sets are frequently small in practice, multiobjective optimization has recently been studied from the view-point of smoothed analysis~\cite{SpielmanT01}. We introduce some notation before describing this work.

%% Given a set of points $S \in \RR^d$, we say that $x \in S$ is not dominated if for all
%% $y \in S \setminus \{x\}$ we have $y \ngtr x$.

\paragraph{Notation.} For positive integer $n$, we denote the set $\{1, 2, \ldots, n\}$ by $[n]$. The multiobjective optimization problems we study have binary variables and linear objective functions. In a general setting, the feasible solution set is an
arbitrary set $\calS \subseteq \{0,1\}^n$. The problem has $d$ linear objective functions $v^{(i)} : \calS \rightarrow \mathbb{R}$, given by
$v^{(i)}(x) = \sum_{j \in [n]} v^{(i)}_j x_j$, for $i \in [d]$, and $(v^{(i)}_1, \ldots, v^{(i)}_n) \in \RR^n$ (so $v^{(i)}$ is also interpreted as an $n$-dimensional vector in the natural way). For convenience, we will assume, unless otherwise specified, that we want to maximize all the objectives, and we will refer to the objectives as profits. This entails no loss of generality. Thus the optimization problem is the following.
\begin{align} \label{mprob}
&\text{maximize } v^{(1)}(x), \ldots, \text{maximize } v^{(d)}(x), \\
&\text{subject to: } x \in \calS.   \nonumber
\end{align}

For the special case of the multiobjective $0$--$1$ knapsack problem we have $d+1$ objectives: One of the objectives will be the weight $w=(w_1, \ldots, w_n)$, which should be minimized, and the other $d$ objectives will be the profits as before: $v^{(i)}=(v^{(i)}_1, \ldots, v^{(i)}_n)$ for $i \in [d]$, and all the entries in $w$ and $v^{(i)}$ come from $[0,1]$.

Let $V$ be the $d \times n$ matrix with rows $v^{(1)}, \dotsc, v^{(d)}$.
We will use the partial order $\preceq$ in $\RR^d$ defined by $x \preceq y$ iff for all $i \in [d]$ we have $x_i \leq y_i$. For $a, b \in \RR^{d}$ we say that $b$ dominates $a$ if $b_i \geq a_i$ for all $i \in [d]$, and for at least one $i \in [d]$, we have strict inequality.
We denote the relation of $b$ dominating $a$ by $b \succ a$.
A solution $x \in \calS$ is said to be Pareto-optimal (or maximal under $\preceq$) if $Vx \nprec Vy$ for all $y \in \calS$. For the knapsack problem, we need to modify the definition of domination appropriately because while we want to maximize the profit, we want to minimize the weight. It will be clear from the context which notion is being used.
For a set of points $X$ in Euclidean space, let $p(X)$ denote the number of Pareto-optima in $X$.

%A random variable taking values in $\mathbb{R}$ is said to have a symmetric distribution if its density function (which we assume exists) is symmetric around the origin.

\paragraph{Smoothed analysis.} For our multiobjective optimization problem \eqref{mprob}, in the worst case the size of the Pareto set can be exponential even for $d=2$ (the bicriteria case). Smoothed analysis is a framework for analysis of algorithms introduced by Spielman and Teng~\cite{SpielmanT01} to explain the fast running time of the Simplex algorithm in practice, despite having exponential running time in the worst case.  Beier and  V{\"o}cking~\cite{BeierV04} studied bicriteria $0$--$1$ knapsack problem under smoothed analysis.
In our context of multiobjective optimization, smoothed analysis would mean that the instance (specified by $V$) is chosen adversarially, but then each entry is independently perturbed according to, say, Gaussian noise with small standard deviation. In fact, Beier and V{\"o}cking~\cite{BeierV04} introduced a stronger notion of smoothed analysis. In one version of their model, each entry of the matrix $V$ is an independent random variable taking values in $[0,1]$ with the restriction that each has probability density function bounded above by $\phi$, for a parameter $\phi \geq 1$. We refer to distributions supported on $[-1,1]$ with probability density bounded above by $\phi$ as $\phi$-semirandom distributions. For more generality, one of the rows of $V$ could be chosen fully adversarially (deterministically). As $\phi$ is increased, the semirandom model becomes more like the worst case model.
With the exception of Theorem~\ref{thm:phi} below, we do not require adversarial choice of a row in $V$.

\paragraph{Previous work.} Beier and V{\"o}cking~\cite{BeierV04} showed that in the above model for the bi-criteria version of the $0$--$1$ knapsack problem with adversarial weights the expected number of Pareto optima is $O(\phi n^4)$; this was improved to $O(\phi n^2)$ by \cite{BeierRV07}.
%These results are actually stronger in the sense that they allow the weights to be chosen completely adversarially.
R{\"o}glin and Teng~\cite{RoglinT09} studied the multiobjective optimization problem in the above framework. They showed that the expected size of the Pareto set with $d$ objectives is of the form $O((\phi n)^{2^{d-2}(d+1)!})$.
Moitra and O'Donnell~\cite{MoitraO11} improved this upper bound to $2\cdot (4 \phi d)^{d(d+1)/2}\cdot n^{2d}$. These authors \cite{RoglinT09, MoitraO11} raised the question of lower bound on the expected number of Pareto optima. Again, these results allow one of the objectives to be chosen adversarially.

An early
average-case lower bound of $\Omega(n^2)$ was proven in \cite{BeierV04RK} for the knapsack problem with a single profit vector. Their
result however required an adversarial choice of exponentially increasing weights.
Recently, Brunsch and R{\"o}glin~\cite{BRoglin11} proved lower bounds of the form 
\[
\Omega_d(\min \{(n\phi)^{(d-\log_2{d})\cdot (1-\Theta(1/\phi))}, 2^{\Theta(n)}\}),
\] 
where $\Omega_d$ means that the constant in the asymptotic notation may depend on $d$. Unfortunately, the instances constructed by them use $\calS$ that does not seem to correspond to natural optimization problems.

\paragraph{Our results.} In this paper we prove lower bounds on the expected number of Pareto optima. Our basic result deals with the case when every entry in the matrix $V$ is chosen independently from a distribution with density symmetric around the origin. Note that we do not require that the distributions be identical: Each entry can have different distribution but we require that the distributions have a density. This generality will in fact be useful in the proof of Theorem~\ref{thm:trees}. Note also that all entries of $V$ are random unlike the results discussed above where one of the objectives is chosen adversarially. This makes our lower bound stronger.

\begin{theorem}[Basic theorem]\label{thm:main}
Suppose that each entry of a $d \times n$ random matrix $V$ is chosen independently according to (not necessarily identical) symmetric distributions with a density. Let $X$ denote the random set $\{Vr : r \in \{0,1\}^n\}$. Then
\begin{align}
\E_V \: p(X) \geq \frac{1}{2^{d-1}} \sum_{k=0}^{d-1} \binom{n-1}{k}.
\end{align}
\end{theorem}
This implies the simpler bound $ \E_V \: p(X) \geq \bigl( \frac{n-1}{2(d-1)}\bigr)^{d-1}$.

We give two proofs of this result. The two proofs have a similar essence, but a somewhat different form. Both proofs relate the
problem at hand to some well-known results in geometry. This connection with geometry is new, and may be useful for future
research. The first proof lower bounds the expected number of Pareto-optima of a point set by the expected number of vertices of its convex hull (up to a constant that depends on $d$ but not on $n$) and then invokes known lower bounds on the expected number of vertices of projections of hypercubes. The second proof gives a characterization of maximality in terms of $0$--$1$ vectors and then relaxes integrality to get a relaxed dual characterization by means of convex separation, which reduces the counting of Pareto-optima to lower bounding the probability that the convex hull of $n$ random points contains the origin. This probability is known exactly by a theorem of Wendel.

Interestingly, our lower bound is basically the same as the expected number of Pareto optima when $2^n$ uniformly random points are chosen from  $[-1,1]^{d}$, which is shown to be $\Theta_d(n^{d-1})$ in several papers~\cite{Bentley, Devroye, Buchta}. This raises the possibility of a closer connection between the two models; such a connection could be useful as the model of uniformly random points is better understood.

The basic theorem above corresponds to the case when the set of feasible solutions $\calS$ is $\{0,1\}^n$. But in many interesting cases $\calS$ is a strict subset of $\{0,1\}^n$: For example, in the multiobjective spanning tree problem $n$ is the number of edges in an underlying network, and $\calS$ is the set of incidence vectors of spanning trees in the network; similarly, for the multiobjective shortest path problem $\calS$ is the set of incidence vectors of $s$--$t$ paths. We can use our basic theorem to prove lower bounds on the size of the Pareto set for such $\calS$. Our technique is pliable enough to give interesting lower bounds for many standard objective functions used in multiobjective optimization (in fact, any standard objective that we tried): Multiobjective shortest path, TSP tour, matching, arborescence, etc. We will illustrate the idea with the multiobjective spanning tree problem on the complete graph. In this problem, we have the complete undirected graph $K_n$ on $n$ vertices as the underlying graph.
Each edge $e$ has a set of profits $v^{(i)}(e) \in [-1, 1]$ for $i \in [d]$.
The set $\calS$ of feasible solutions is given by the incidence vectors of spanning trees of $K_n$. Notice that the feasible set here lives in $\{0,1\}^{n \choose 2}$ and not in $\{0,1\}^n$.

\begin{theorem} \label{thm:trees}
In the $d$ objective maximum spanning tree problem on $K_n$ there exists a choice of $4$-semirandom distributions such that the expected number of Pareto-optimal spanning trees is at least $(\frac{n-3}{2(d-1)})^{d-1}$.
\end{theorem}

The proof of this theorem utilizes the full power of Theorem~\ref{thm:main}, namely the ability to choose different symmetric
distributions.

In our basic theorem above, Theorem \ref{thm:main}, we required the distributions to be symmetric, and therefore that theorem does not apply to the $0$--$1$ knapsack problem where all profits and weights are non-negative. With a slight loss in the lower bound we also get a lower bound for this case.
In the multiobjective $0$--$1$ knapsack problem we have $d$ objectives $v^{(i)}$ for $i \in [d]$ called profits and an additional objective $w$ called weight. Components of $p^{(i)}$ and $w$ are all chosen from $[0,1]$. We want to maximize the profits and minimize the weight, and so the definitions of domination and Pareto-optimality are accordingly modified.

\begin{theorem} \label{thm:basicknapsack}
For the multiobjective $0$--$1$ knapsack problem where all the weight components are $1$ and profit components are chosen uniformly at random from $[0,1]$, the expected number of Pareto optima is $\Omega_d(n^{d-1.5})$.
\end{theorem}

Theorems \ref{thm:main} or \ref{thm:basicknapsack} (depending on whether one wants a bound for non-negative or unrestricted weights and profits) can be used in a simple way as the base case of the argument with $d+1$ objectives in \cite[Section 3]{BRoglin11} to give the following improved lower bound on the expected number of Pareto optima when the profits are $\phi$-semirandom (actually, uniform in carefully chosen intervals of length at least $1/\phi$):
\begin{theorem} \label{thm:phi}
For any fixed $d \geq 2$ (so that the constants in asymptotic notation may depend on $d$) and for $n \in \NN$ and $\phi > 1$ there exist
\begin{enumerate}
\item weights $w_1, \dotsc, w_n \geq 0$,
\item intervals $[a_{ij}, b_{ij}] \subseteq [0,1]$, $i \in [d], j \in [n]$ of length at least $1/\phi$ and with $a_{ij} \geq 0$, and
\item a set $\calS \subseteq \{0,1\}^n$
\end{enumerate}
such that if profits $v^{(i)}_{j}$ are chosen independently and uniformly at random in $[a_{ij}, b_{ij}]$, then the expected number of Pareto-optimal solutions of the $(d+1)$-dimensional knapsack problem with solution set $\calS$ is at least
\[
\min \{ \Omega_d( n^{d-1.5} \phi^{(d-\log d) (1-\Theta(1/\phi))}), 2^{\Theta(n)} \}.
\]
\end{theorem}

For general multiobjective optimization (basically without the restriction of entries being non-negative) the exponent of $n$ becomes exactly $d$.

The technique of \cite{BRoglin11} requires $\calS$ to be chosen adversarially, and so this is the case for Theorem~\ref{thm:phi} above as well.  To our knowledge, no non-trivial lower bounds were known before our work for natural choices of $\calS$. This is addressed by our Theorems~\ref{thm:main} ($\calS = \{0,1\}^n$) and \ref{thm:trees} ($\calS$ is the set of spanning trees of the complete graph) above, though these Theorems are for a small constant value of $\phi$, and therefore do not clarify what the dependence of $\phi$ should be.

Very recently, Brunsch and R{\"o}glin improved the induction step of their lower bound \cite{BRoglin112}. Combining their improved
result with our result yields the lower bound of $\min\{\Omega_d(n^{d-1.5} \phi^d), 2^{\Theta(n)}\}$.
%we are unable to get a dependence on $\phi$ in these cases.
%% Given a set of points $S \in \RR^d$, we say that $x \in S$ is not dominated if for all
%% $y \in S \setminus \{x\}$ we have $y \ngtr x$.

\vspace{-0.1in}
\section{The basic theorem}
\vspace{-0.1in}
In this section we prove Theorem~\ref{thm:main}.
We will include two proofs that, while in essence the same, emphasize the geometric and algebraic views, respectively.  Also the second proof is more self-contained. It is perhaps worth mentioning that we first discovered the second proof, and in the course of writing the present paper we found the ideas and known results that could be combined to get a more geometric proof.
%It is perhaps worth mentioning that we first discovered the second proof, and in the course of writing the present paper we found results in the literature that could be combined to get a shorter proof.
\vspace{-0.05in}
\subsection{First proof}
\vspace{-0.05in}
\begin{proof}
The convex hull of $X$ is a random polytope, a zonotope actually, that is, a linear image of a hypercube or, equivalently, a Minkowski sum of segments.
By Theorem \ref{thm:bentley} \cite{Bentley}, every vertex is maximal under our partial order $\preceq$ for at least one of the $2^d$ reflections involving coordinate hyperplanes. That is
\begin{equation}\label{equ:bentleyetal}
\card{\text{vertices of $\conv{X}$}} \leq \sum_{\eps \in \{-1, 1\}^d} p(\text{$X$ with coordinates of points flipped by signs in $\eps$})
\end{equation}
Our symmetry assumption followed by \eqref{equ:bentleyetal} implies\footnote{This idea is from \cite{Bentley}. It is used there in the opposite direction, that is, to get upper bounds on the expected number of vertices from upper bounds on the expected number of maximal points.}
\begin{align*}
\E_V(p(X)) &= \frac{1}{2^d} \cdot \sum_{\eps \in \{-1, 1\}^d} \E \; p(\text{$X$ with coordinates of points flipped by signs in $\eps$}) \\
&\geq \frac{1}{2^d} \cdot \E \: \card{\text{vertices of $\conv{X}$}}
\end{align*}
It is known \cite[Theorem 1.8]{donoho} that for $V$ with columns in general position (that is, any $d$ columns are linearly independent, which happens almost surely in our case) the number of vertices is equal to the maximal number of vertices of a $d$-dimensional zonotope formed as the sum of $n$ segments \cite[31.1.1]{goodman2004handbook}. That is, almost surely:
\[
\card{\text{vertices of $\conv{X}$}} = 2 \sum_{k=0}^{d-1} \binom{n-1}{k}.
\]
The claimed bound follows.
\end{proof}
We used the following result:
\begin{theorem}[\cite{Bentley}, {\cite[Theorem 4.7]{MR805539}}]\label{thm:bentley}
Let $P$ be a finite subset of $\RR^d$. A vertex of the convex hull of $P$ is maximal under $\preceq$ in at least one of the $2^d$ assignments of $d$ signs $+$ and $-$ to each of the coordinates of the points of $P$.
\end{theorem}
%\begin{proofidea}
%If a point is not maximal for any sign assignment then it
%\end{proofidea}
\vspace{-0.05in}
\subsection{Second proof}
\vspace{-0.05in}
Some more definitions before getting into the proof:
Set $\RR_+ = \{x \in \RR : x \geq 0\}$ and $\RR_- = \{x \in \RR : x \leq 0\}$.
For $\epsilon \in \{-1, 1\}^\dimd$,
the orthant associated with $\epsilon$ is $\{(\epsilon_1 x_1, \ldots, \epsilon_\dimd x_\dimd)\::\:
(x_1, \ldots, x_\dimd) \in \RR_+^\dimd\}$. In particular, if $\epsilon$ is the all $1$'s vector then we call its associated orthant the positive orthant, and if $\epsilon$ is the all $-1$'s vector then we call its orthant the negative orthant. For a finite set of points $P = \{p_1, \dotsc, p_k\} \subseteq \RR^d$, the conic hull is denoted $\cone(P) = \{ \sum_{i=1}^k \alpha_i p_i \suchthat \alpha_i \geq 0\}$ (note that the conic hull is always convex).

\begin{proof}
By linearity of expectation
$\E \: p(X) = \sum_r \Pr[Vr \text{ maximal}].$
Notice that $\Pr[Vr \text{ maximal}]$ does not depend on $r$, so we can
write
$\E \: p(X) = 2^n \Pr[V1 \text{ maximal}]$.
%(This only uses that the columns of are independent and have a ``1-unconditional'' distribution. It may be helpful here to think of $\{-1, 1\}^n$ instead of $\{0,1\}^n$, the notion of corresponding point being Pareto optimum is invariant)

For the rest of the proof we will focus on finding a lower bound on this last
probability.
To understand this probability we first rewrite the event $[V1 \text{ maximal}]$ in
terms of a different event via easy intermediate steps:
\begin{align*}
[V1 \text{ maximal}] = [Vr \nsucc V1, \;\; \forall r \in \{0,1\}^n]
= [0 \nsucc V(1-r), \;\; \forall r \in \{0,1\}^n]
= [0 \nsucc V r, \;\; \forall r \in \{0,1\}^n].
\end{align*}

Now we have $
\Pr[0 \nsucc V r, \;\; \forall r \in \{0,1\}^n] \geq
\Pr[0 \nsucc Vr, \;\; \forall r \in [0,1]^n].$

Event $[0 \nsucc Vr, \;\; \forall r \in [0,1]^n]$ is the same as the event
$[\cone(v_1, \ldots, v_n) \cap \RR_-^{\dimd} = \{0\}]$.  That is to say, the cone
generated by the non-negative linear combinations of $v_1, \ldots, v_n$ does not
have a point distinct from the origin that lies in the negative orthant.

By the separability property of convex sets (Hahn-Banach theorem) we have
that there exists a hyperplane $H = \{x \in \RR^d : \dotp{u}{x}=0\}$ separating
$\cone(v_1, \ldots, v_n)$ and $\RR_-^{d}$. That is, there exists $u \in \RR_+^d \setminus \{0\}$ such that $\cone(v_1, \ldots, v_n)\cdot u \geq 0$ and this implies
\[
\Pr[\cone(v_1, \ldots, v_n) \cap \RR_-^{\dimd} = \{0\}] = \Pr[\exists u \in \RR_+^\dimd \setminus \{0\} : \cone(v_1, \ldots, v_n)\cdot u \geq 0].
\]
%The normal $u$ to $H$ can be chosen to lie in the positive orthant. {\bf Navin: this should be proven.}
Now
\begin{align*}
\Pr[\cone(v_1, &\ldots, v_n) \text{ in a halfspace}]\\
&\leq \sum_{\epsilon \in \{-1, 1\}^\dimd} \Pr[\cone(v_1, \ldots, v_n) \text{ in a halfspace with inner normal in orthant $\epsilon$}] \\
&= 2^\dimd \Pr[\exists u \in \RR_+^\dimd\setminus \{0\} : \cone(v_1, \ldots, v_n)\cdot u \geq 0].
\end{align*}
Clearly, we have
\begin{align*}
[\cone(v_1, \ldots, v_n) \text{ in a halfspace}] = [v_1, \ldots, v_n \text{ in a halfspace}].
\end{align*}
Theorem \ref{thm:wendel} and the fact that the distribution of $v_i$ is centrally symmetric and assigns measure zero to every hyperplane through 0 imply
\begin{align*}
\Pr_V[v_1, \ldots, v_n \text{ in a halfspace}]
%&= \Pr_{\eps \in \{-1, 1\}^n, V} [\eps_1 v_1, \ldots, \eps_n v_n \text{ in a halfspace}] \\
%& \geq \frac{2 \binom{n}{d-1}}{2^n}
= \Pr_V [0 \notin \conv\{v_1, \dotsc, v_n\}]
= \frac{1}{2^{n-1}} \sum_{k=0}^{\dimdm} \binom{n-1}{k}.
%\\
%&\geq \frac{1}{2^{n-1}} \left( \frac{n-1}{d-1}\right)^{d-1}.
\end{align*}
We conclude:
%\begin{align*}
%\E f(V) &\geq  \left( \frac{n-1}{2(d-1)}\right)^{d-1}.
%\end{align*}
\begin{align*}
\E \: p(X) &\geq  \frac{1}{2^{\dimdm}} \sum_{k=0}^{\dimdm} \binom{n-1}{k}.
\end{align*}
\end{proof}

We used the following result by Wendel:
\begin{theorem}[\cite{MR0146858}, {\cite[Theorem 8.2.1]{schneider}}]\label{thm:wendel}
If $X_1, \dotsc, X_n$ are independent random points in $\RR^d$ whose distributions are symmetric with respect to 0 and such that with probability 1 all subsets of size $d$ are linearly independent, then
\[
\Pr[0 \notin \conv\{X_1, \dotsc, X_n \}] = \frac{1}{2^{n-1}} \sum_{k=0}^{d-1} \binom{n-1}{k}.
\]
\end{theorem}
The linear independence condition holds in particular under the simpler assumption that no hyperplane through the origin is assigned positive probability by any of the $n$ distributions. For example, it holds when the points are i.i.d. at random from the unit sphere.
%\begin{theorem}[\cite{MR0146858}, {\cite[Theorem 8.2.1]{schneider}}]\label{thm:wendel}
%If $X_1, \dotsc, X_n$ are i.i.d. random points in $\RR^d$ whose distribution is symmetric with respect to 0 and assigns measure zero to every hyperplane through 0, then
%\[
%\Pr[0 \notin \conv\{X_1, \dotsc, X_n \}] = \frac{1}{2^{n-1}} \sum_{k=0}^{d-1} \binom{n-1}{k}.
%\]
%\end{theorem}

The following easy corollary of Theorem \ref{thm:main} will be useful later.
\begin{cor} \label{cor:-11}
Theorem~\ref{thm:main} holds also when the feasible set $\calS = \{-1,1\}^n$.
\end{cor}
\begin{cor} \label{cor:S}
Under the assumptions of Theorem \ref{thm:main} and when the set of feasible solutions is $\calS \subseteq \{0, 1\}^n$ we have
\begin{align}
\E_V \: p(X) \geq \frac{\card{\calS}}{2^{n+d-1}} \sum_{k=0}^{d-1} \binom{n-1}{k}.
\end{align}
\end{cor}
\begin{proof}
For any given $r \in \calS$, the probability that it is Pareto-optimal in the current instance with solution set restricted to $\calS$ is at least the probability that it is Pareto-optimal in the instance with solution set $\{0,1\}^n$. By the symmetry of the distribution this probability is independent of $r$ and by Theorem \ref{thm:main} it is at least
\begin{align*}
\frac{1}{2^{n+d-1}} \sum_{k=0}^{d-1} \binom{n-1}{k}.
\end{align*}
Linearity of expectation completes the proof.
\end{proof}
\section{Lower bound for multiobjective maximum spanning trees}
In this section we show that our basic result can be used to derive similar lower bounds for $\calS$ other than those encountered earlier in this paper. We illustrate this for the case of multiobjective maximum spanning tree problem on the complete graph; for this problem, $\calS$ is the set of incidence vectors of all spanning trees on $n$ vertices. The idea of the proof is simple: ``Embed'' the instance of the basic problem into the instance of the problem at hand. The proof requires the full power of Theorem~\ref{thm:main}.  %% As mentioned before, similar bounds are achievable using similar ideas in all cases of interest that we know of, e.g., multiobjective shortest path, TSP tour, arborescence, matching.
It is worth noting that the direct use of Cor.~\ref{cor:S} does not provide any useful bound for the case of spanning trees. The proof below is easily modified so that all profits are chosen from intervals of non-negative numbers.

We now prove Theorem~\ref{thm:trees}.

\begin{proof}
The idea of the proof is to embed an instance of the case $\calS = \{0,1\}^{n-2}$ of the basic theorem into the tree instance. We now describe our tree instance. We identify a subgraph $G$ of $K_n$ (the complete graph on $n$ vertices): The vertex set of $G$ is the same as the vertex set of $K_n$, which for convenience we denote by $\{s, t, u_1, u_2, \ldots, u_{n-2}\}$. The edge set of $G$ consists of the edge $(s,t)$, and edges $(s,u_j), (t,u_j)$. Thus, $G$ consists of $2n-3$ edges. Now we choose the distribution of the profits for each edge of $K_n$. For edges outside $G$, the distribution for all profits is identical and it is
simply the uniform distribution on $[-1,-1/2]$. For edge $(s,t)$, the distribution is uniform over $[1/2,1]$. And for all other edges it is uniform over $[-1/2,1/2]$.
Let $\mathcal{T}$ denote the set of spanning trees which include edge $(s,t)$, and for every other vertex $u_j$, exactly one of $(s, u_j)$ and $(t, u_j)$.
Clearly $|\mathcal{T}| = 2^{n-2}$.
The result of the above choices of distributions is that all the Pareto-optimal spanning trees come from $\mathcal{T}$:
\begin{claim} \label{claim:sts}
For any choices of profits from the intervals as specified above, if a tree $T$ is Pareto-optimal then $T \in \mathcal{T}$ .
\end{claim}
\begin{proof}
Fix any choice of profits as above. Suppose that a tree $T'$ is Pareto-optimal but $T' \notin \mathcal{T}$. Then (1) either $T'$ has an edge $e$ outside $E(G)$, or (2) all its edges are from $E(G)$ but it does not use edge $(s,t)$. In case (1), remove the edges from $T'$ that are not in $E(G)$, and then complete the remaining disconnected graph to a spanning tree using edges from $E(G)$. Clearly, the resulting tree is heavier than $T'$ in each of the $d$ weights. In case (2), add edge $(s,t)$ to $T'$, and from the resulting cycle remove some edge other than $(s,t)$. Again, the resulting tree is heavier than $T'$ in each of the $d$ weights.
\end{proof}
In the rest of the proof, $i$ will range over $[d]$.
The $i$'th profit of a spanning tree $T \in \mathcal{T}$, which we will denote by $v^{(i)}(T)$, can be written as follows
\begin{align*}
v^{(i)}(T) = v^{(i)}(st) + \sum_{j=1}^{n-2} (v^{(i)}(s,u_j) x_j + v^{(i)}(t,u_j) (1-x_j)),
\end{align*}
where $x_j = x_j(T) = 1$ if edge $(s, u_j)$ is in the tree and $x_j = 0$ otherwise. We have
\begin{align*}
(v^{(i)}(s,u_j) x_j + v^{(i)}(t,u_j) (1-x_j)) = \frac{v^{(i)}(s,u_j)+v^{(i)}(t,u_j)}{2} + (v^{(i)}(s,u_j)-v^{(i)}(t,u_j))(x_j-\frac{1}{2}).
\end{align*}
Now, to compute the lower bound on the expected size of the Pareto set we reveal the $v$'s in two steps: First we reveal $(v^{(i)}(s,u_j)+v^{(i)}(t,u_j))$ for all $u_j$. Then the conditional distribution of each $(v^{(i)}(s,u_j)-v^{(i)}(t,u_j))$ is symmetric (but can be different for different $i$). Thus the $i$'th profit of $T \in \mathcal{T}$ is $v^{(i)}(T) = \sum_{i \in [n-2]} (v^{(i)}(s,u_j)-v^{(i)}(t,u_j)) (x_j-1/2) + A^{(i)}$, where $A^{(i)} =  v^{(i)}(s,t) + \sum_{j \in [n-2]}\frac{v^{(i)}(s,u_j)+v^{(i)}(t,u_j)}{2}$. Since $A^{(i)}$ is common to all trees, only the first sum in the profit matters in determining Pareto-optimality. Now we are in the situation dealt with by Cor.~\ref{cor:-11}: For each fixing of $(v^{(i)}(s,u_j)+v^{(i)}(t,u_j))$, we get an instance of Cor.~\ref{cor:-11}, and thus a lower bound of $(\frac{n-3}{2(d-1)})^{d-1}$. Since this holds for each fixing of $(v^{(i)}(s,u_j)+v^{(i)}(t,u_j))$, we get that the same lower bound holds for the expectation without conditioning.
\end{proof}

\section{$0$-$1$ Knapsack}
We prove Theorem~\ref{thm:basicknapsack}.
\begin{proof}
To show our lower bound we will use the obvious one-to-one map between our basic problem with $d$ objectives and the profits of the knapsack problem: Let $v^{(1)}, \ldots, v^{(d)}$ be an instance of our basic problem with all the $v^{(i)}_j$ being chosen uniformly at random from $[-1/2,1/2]$. Now the profits $p$ are obtained from the $v$'s in the natural way: $p^{(i)}_j = v^{(i)}_j + 1/2$. In general, the set of Pareto optima for these two problems (the basic problem instance and its corresponding knapsack instance) are not the same.
We will focus instead on the better behaved set $\calS \subseteq \{0,1\}^n$ of solutions having exactly $\floor{n/2}$ ones. From Corollary~\ref{cor:S} we get that, in the basic problem restricted to $\calS$, the expected number of Pareto optima is at least $\Omega_d(n^{d-1.5})$ (using the well-known approximation ${n \choose \floor{n/2}} =\Theta( 2^n/\sqrt{n})$).

Now we claim that if $x \in \calS$ is Pareto-optimal in the restricted basic problem, then it is also Pareto-optimal in the corresponding (unrestricted) knapsack problem. Let $y \in \{0,1\}^n$ be different from $x$. There are two cases: If $y$ has more than $\floor{n/2}$ ones, then it cannot dominate $x$, as $y$ has a strictly higher weight (recall that all the weights are $1$). If $y$ has at most $\floor{n/2}$ ones, then enlarge this solution arbitrarily to a solution $y' \succeq y$ with exactly $\floor{n/2}$ ones. The maximality of $x$ implies that $y'$ is worse in some profit, and so is $y$, as the profits are non-negative.
%We know from Cor.~\ref{cor:S} that in the basic problem, the expected number of Pareto optima from the set ${[n] \choose n/2}$ (i.e., the feasible solutions with precisely $n/2$ entries with value $1$) is at least $\Omega_d(n^{d-1.5})$ (using the well-known approximation ${n \choose n/2} \approx 2^n/\sqrt{n}$).
%
%Now we claim that if $x \in {[n] \choose n/2}$ is Pareto-optimal in the basic problem, then it is also Pareto-optimal in the corresponding knapsack problem. To this end consider two cases: (1) If $v(x) \nless v(y)$ for $w(x) < w(y)$, then $(p,w)(x) \nless (p,w)(y)$ simply because $w(x) < w(y)$. (2) If $v(x) \nless v(y)$ for $w(x) > w(y)$, then again $p(x) \nless p(y)$: Let $i \in [d-1]$ be such that
%$v^{(i)}(x) > v^{(i)}(y)$; such an $i$ exists because of our assumption. Now $p^{(i)}(x) = v^{(i)}(x)+n/2$
%and $p^{(i)}(y) = v^{(i)}(y) + w(y) < v^{(i)}(y) + n/2$. Thus $p^{(i)}(x) > p^{(i)}(y)$, which gives $p(x) \nless p(y)$.
\end{proof}
\section{Improved lower bound in the semi-random model}
We prove Theorem~\ref{thm:phi}.
\begin{proof}
We only describe the differences with the argument in the proof of
\cite[Theorem 8]{BRoglin11}.
As given by Theorem \ref{thm:basicknapsack} (but scaling the profits by $1/\phi$, which does not change the set of Pareto-optima), we start with a distribution on knapsack instances with $d$ profits and $n_p$ objects (to be determined later) having unit weights and profits uniformly distributed in $[0, 1/\phi]$, and expected number of Pareto-optima at least $\Omega(n^{d-1.5})$. We use $n_q$ (to be determined later) ``cloning steps''. Each step introduces $d$ new objects while multiplying the number of Pareto-optima by at least $2^d/d$. As in \cite{BRoglin11}, objects used by the splitting step can have profits that are larger than 1, therefore they are split into many objects with profits distributed in $[0,1]$, and a suitable choice of the set $\calS$ ensures that objects representing the splitted version of another behave as a group.

A simple modification of the argument leading to \cite[Corollary 11]{BRoglin11}, using our base case with $n_p$ objects described in the previous paragraph instead of their base case with 1 object, implies that the expected number of Pareto-optima of the constructed instance is $\Omega(n_p^{d-1.5} (2^d/d)^{n_q})$.

Now we need to choose values of $n_p$ and $n_q$ to get a bound in terms of $n$. By \cite[Lemma 11]{BRoglin11}, the total number of objects is at most
\begin{equation}\label{equ:objects}
n_p + dn_q + \frac{2 d^2}{\phi -d} \left( \frac{2\phi}{\phi - d} \right)^{n_q}
\end{equation}
We choose $n_q$ so that the second term is no more than $n/4$ for $n$ and $\phi$ sufficiently large. Such a choice of $n_q$ is given by $n_q = \floor{\hat n_q} $ with
\[
\hat n_q = \frac{\log \phi}{\log \frac{2\phi}{\phi -d}}
\]
when $4 d^2 \leq n/4$ and $\phi \geq 2d$.

Clearly there can be no more than $2^n$ Pareto optima, and therefore there must be a point where increasing $\phi$ does not increase the lower bound. Say, for
\begin{equation}\label{equ:phi}
\phi \leq \left( \frac{2 \phi}{\phi-d}\right)^{n/2d}
\end{equation}
we have that the second term in \eqref{equ:objects} is no more than $n/2$. Finally, choosing $n_p = \floor{n/2}$ ensures that \eqref{equ:objects} is at most $n$.

As explained in the first paragraph, the expected number of Pareto-optima of the whole construction is at least
\begin{align*}
\Omega\left( n_p^{d-1.5} \left(\frac{2^d}{d}\right)^{n_q}\right) \geq \Omega \left( n_p^{d-1.5} \left(\frac{2^d}{d}\right)^{\hat n_q}\right) \geq \Omega( n^{d-1.5} \phi^{(d - \log d)(1- \Theta(1/\phi))}).
\end{align*}
When $\phi$ violates \eqref{equ:phi}, we construct the same instance as above with maximum density equal to the unique $\hat \phi$ satisfying $\hat \phi = \left( \frac{2 \hat \phi}{\hat \phi-d}\right)^{n/2d}$ ($\hat \phi$ is about $2^{n/2d}$). We get $\hat n^q = n/2d$ and, as before, the expected number of Pareto-optima is at least
\[
\Omega\left( n_p^{d-1.5} \left(\frac{2^d}{d}\right)^{n_q}\right) \geq \Omega\left( n^{d-1.5} \left(\frac{2^d}{d}\right)^{n/2d}\right) \geq \Omega(2^{\Theta(n)})
\]
\end{proof}
\section{Discussion and Conclusion}
We proved lower bounds for the average and smoothed number of Pareto optima by introducing geometric arguments to this setting. Our lower bound is of the form $\Omega(n^{d-1})$, ignoring the dependence on $\phi$. The best upper bound we know, even for $\phi =1$, is that of Moitra and O'Donnell~\cite{MoitraO11} which is of the form $O(n^{2d-2})$, again ignoring the dependence on $\phi$. Thus there is a gap between the upper and lower bounds. As mentioned before, the number of Pareto optima for the case when $2^n$ points are chosen uniformly at random from $[-1,1]^d$ is $\Theta_d(n^{d-1})$.
%If this is any indication then our lower bound is closer to the truth.

%% An upper bound proven in the $\phi$-semirandom model is a stronger result than the same upper bound in the original perturbation
%% model of Spielman--Teng for smoothed analysis. The reverse is true for lower bounds. Thus, perhaps more informative and useful
%% lower bounds for the number of Pareto optima are those that hold under the perturbation model and do not use the unrealistic power
%% of the $\phi$-semirandom distributions. More precisely, lower bound in the perturbation model would exhibit an input
%% such that random perturbations of it have large number of Pareto optima. Lower bounds proven here, specifically
%% Theorems \ref{thm:trees} and \ref{thm:basicknapsack}, are of this type.

%% What constitutes a lower bound in
%% this setting deserves a bit more attention. The simplest lower bound would be one where one exhibits a hard input
%% whose perturbation on average produces large number of Pareto optima. Lower bounds proven here are of this type.
%% One could ask for something more ambitious: Is it true that for \emph{most} inputs
%%  average number of Pareto optima is large after perturbations. Such lower bounds would be more in the spirit of smoothed upper bounds than exhibiting a single hard input.

%%But even so, there are gaps in our understanding of lower bounds that we leave unanswered:

Do lower bounds similar to ours hold for any sufficiently large feasible set $\calS$? Our techniques can show this
for natural objectives, but require arguments tailored to the specific objective.
 It is desirable to have a general lower bound technique that works for all sufficiently large $\calS$. Also, in smoothed lower bounds, to get a good dependence on $\phi$ we need to use the technique of $\cite{BRoglin11}$, which requires a very special choice of $\calS$. So, a more general question is whether we can prove lower bounds with strong dependence on $\phi$ for all sufficiently large $\calS$.

We now briefly discuss some difficulties in proving lower bounds for general $\calS$. One approach
to this end is to show a lower bound on the expected size of the Pareto set that depends only on $\card{\calS}$, $n$ and $d$.
Our general technique was to first reduce the problem to lower bounding the expected number of vertices in the
projection of the convex hull of the points in $\calS$ to a random subspace of dimension $d$. %%  The number of vertices in the projection is an
%% interesting quantity for various reasons and has been studied; see \cite{SchneiderSurvey} for references.
A special distribution
which is instructive to consider here, and also interesting in its own right, is given by the case when we
project to a $d$-dimensional space chosen uniformly at random. The expected number of vertices in the projection has been
studied for the special cases of the simplex, the cube, and
the crosspolytope (see Schneider~\cite{SchneiderSurvey}). But understanding this number for arbitrary 0/1-polytopes seems
difficult.
%% One may further restrict to the case where the objectives are standard Gaussian vectors. In this case, after using the connection between extreme points and Pareto optima (from \cite{Bentley}, used in Section xxx here), the problem reduces to lower bounding the expected number of vertices of a random projection of an $n$-polytope, the convex hull of $\calS$, onto a $d$-dimensional subspace. This problem has received some attention for special polytopes such as the simplex [xxx], the cube [xxx] and the cross-polytope [xxx], but understanding it for arbitrary 0--1 polytopes seems difficult.
When the subspace to be projected to is of dimension $n-1$, we can write the expected
number of vertices in the projection as $C \cdot \sum_{v \in V} a(v)$, where $a(v)$ is the solid angle of the cone polar to the
tangent cone at vertex $v$, and $C$ is a constant depending on $n$. (Suitable generalizations of this formula are easy
to obtain for projection to dimensions smaller than $n-1$, but the case of dimension $n-1$ is sufficient
for our purpose here.) This captures the intuitive fact that if the polytope is very pointy at vertex $v$, then $v$ is more
likely to be a vertex in the convex hull.
%% It can be shown that, given a polytope, the expected number of vertices in the random projection can be expressed in terms of spherical intrinsic volumes of the internal and external angles of the polytope [xxx].
%% For example, the expected number of vertices in the projection onto a random hyperplane is proportional to the sum of the interior angles [xxx].
It is natural to ask: Given $k$, what is the $\calS \subseteq \{0,1\}^n$ with $\card{\calS}= k$ that minimizes this expectation?
Intuitively, the sum of angles $a(v)$ could be minimized when the vertices are close together, as in a Hamming ball.
Note the high-level similarity of the problem at hand to the edge-isoperimetric inequality for the Boolean cube. Unfortnately, our numerical experiments show that this is not the case: Hamming balls are not the minimizers of the expected number of vertices of a random projection.

%% For random projections of simplices it is known that...
%% This together with the result for the cube suggests that the minimum expected size of the projection might behave like
%% $\Omega_d(\log^d|\calS|)$. This however is far from clear.

\section{Acknowledgements}
We thank Yusu Wang for helpful discussions.

\bibliography{pareto}
\bibliographystyle{abbrv}    % bibliography using BibTex format
\end{document}